\documentclass[12pt,onepage,oneside]{article}
\usepackage{fullpage}
\usepackage{graphicx}
\usepackage{epstopdf}
\usepackage{etex}
\usepackage{tocloft}
\usepackage{booktabs}
\usepackage{natbib}
\usepackage{fancyhdr}
\usepackage{hyperref}
\usepackage{amsthm}
\usepackage{amsmath}
\usepackage{cleveref}
\usepackage{bbm}
\usepackage{multirow}
\usepackage{lscape}
\usepackage{bbm}
\usepackage{bm}
\usepackage{color}
\usepackage{setspace,caption}
\usepackage{subcaption}
\usepackage{ifthen}
\usepackage[latin1]{inputenc}
\usepackage{titletoc}
\usepackage{eso-pic}
\usepackage{float}

\usepackage{amssymb}
\usepackage{xcolor}

\newtheorem{assumption}{Assumption}
\newtheorem{prop}{Proposition}

\def \R {\mathbb  R}

\def \E {\mathbb  E} 

\usepackage{color,soul}
\definecolor{lightblue}{rgb}{.80,.95,1}
\sethlcolor{lightblue}
\setulcolor{red}

\usepackage[margin=1in]{geometry}

\usepackage{times}

\begin{document}

\title{The Dynamic Persistence of Economic Shocks}

\author{%
Jozef {\sc Barun\'{i}k}$^{\rm a,b}$\thanks{Corresponding author, Tel. +420 (776) 259 273, Email address: barunik@fsv.cuni.cz}, and
Luk\'{a}\v{s} {\sc V\'{a}cha}$^{\rm b,a}$
\vspace{5mm} \\
 \small $^{\rm a}$ Institute of Economic Studies, Charles University, \vspace{-0.5mm}\\  \
 \small Opletalova 26, 110 00, Prague, Czech Republic \vspace{3mm} \\
 \small $^{\rm b}$ The Czech Academy of Sciences, Institute of Information Theory and Automation \vspace{-0.5mm}\\
  \small Pod Vodarenskou Vezi 4, 182 00, Prague, Czech Republic}

\maketitle
\begin{abstract}
{\normalsize We propose a novel framework for modeling time-varying persistence in economic time series, allowing for smoothly evolving heterogeneity in shock dynamics. We leverage localized regression techniques to flexibly identify changes in persistence over time, offering a data-driven alternative to traditional parametric models. We applied this methodology to U.S. inflation and stock market volatility data and found substantial persistence variations that align with key macroeconomic events and market conditions. The results reveal previously undetected pockets of predictability and provide significant increases in out-of-sample forecast accuracy. These findings have important implications for economic modeling, forecasting, and policy analysis.}

\end{abstract}

\bigskip

\noindent {\normalsize \textbf{Acknowledgments}: We are grateful to the editor, Peter Hull, and three anonymous reviewers for their useful comments and suggestions, which greatly improved the paper. We also thank Rainer von Sachs, Federico Severino, Roman Liesenfeld, Nikolas Hautsch, Christian Hafner, Wolfgang H\"{a}rdle, and Lubos Hanus for valuable discussions and comments. We appreciate the insightful comments from numerous seminar presentations, such as the Recent Advances in Econometrics (2023, Louvain), the 2021 and 2022 STAT of ML conference in Prague, and the 15${\text{th}}$ International Conference on Computational and Financial Econometrics. The support of the Czech Science Foundation within the project 24-11555S is gratefully acknowledged. To estimate the proposed quantities, we provide the \texttt{tvPersistence.jl} package in \textsf{JULIA} with the excellent research assistance of Ji\v{r}\'i Mikulenka. The package is available at \url{https://github.com/barunik/tvPersistence.jl}.\break
\noindent \textbf{Disclosure Statement:} Jozef Barun\'{i}k and Lukas Vacha have nothing to disclose.}


\newpage
\section{Introduction}

Macroeconomic and financial variables have varied considerably over the past decades \citep{primiceri2005time,justiniano2008time,bekierman2018forecasting,chen2018nonparametric}, as both stable and uncertain periods associated with different economic states have been driven by different shocks. Moreover, many authors argue that these variables are driven by shocks that influence their future value with heterogeneous degrees of persistence \citep{dew2013asset,bandi2021,bandi2022spectral,bandi2017business}.\footnote{We use the term persistence to capture a property of a time series that is closely related to its autocorrelation structure. In particular, the degree of persistence can be used to accurately describe how a shock affects the time series. A low degree of persistence indicates shocks of a transitory nature that force the time series to return to its mean path. In contrast, when shocks push the time series away from the mean path, they are said to be highly persistent. A shock tends to persist for a long time.} A possibly nonlinear combination of transitory and persistent responses to shocks results in time series with heterogeneous persistent structures that remain hidden to the observer using traditional methods. Each of these challenges severely limits progress in modeling and forecasting economic data, as researchers work with models that aggregate one of these features.

Identifying and modeling such heterogeneity is crucial. Persistence directly informs our understanding of economic dynamics and has material implications for forecasting, structural analysis, and policy design. However, despite extensive work on unit roots \citep{evans1981testing,nelson1982trends,perron1991continuous}, structural breaks \citep{perron1989great}, and more complicated long memory or fractionally integrated structures that exhibit considerable persistence over time without being nonstationary \citep{hassler1995long,baillie1996analysing}, there remains no consensus on how best to identify or model the time-varying nature of persistence. In existing approaches, strong parametric assumptions are often imposed; moreover, such methods fail to account for smooth shifts in dependent structures, limiting both interpretability and forecasting performance.

To identify and explore local persistent structures that are useful for modeling and forecasting, we propose a novel representation for nonstationary time series data, providing a precise characterization of how shocks of differing durations contribute to the overall dynamics of a time series at each point in time. In addition, we propose a novel forecasting approach using time-varying coefficient regressions with heterogeneously persistent components, which can be used to identify ``pockets of predictability'' that occur owing to shocks in the time series, which are characterized by a local heterogeneously persistent nature. By providing a flexible, interpretable, and empirically tractable model of time-varying persistence, this paper contributes to the literature on time series decomposition, economic forecasting, and the evolving structure of macro-financial shocks. Our findings suggest that accommodating heterogeneous and shifting persistent structures is essential to capture dynamics that are often missed by traditional models and to exploit information that is otherwise hidden within the data.

Our work is closely related to recent works in the literature that suggested representing covariance stationary time series as linear combinations of orthogonal components carrying information about heterogeneous cycles of various lengths \citep{bandi2019,ortu2020}. While these methods are particularly well suited to studying heterogeneously persistent structures in time series data, the stable and uncertain times associated with different economic states imply that responses to shocks vary over time, and these shocks are difficult to identify under the assumption of stationary data. 
Thus, the localization of persistent structures is a new approach to modeling and forecasting. A model that allows persistent structures to change smoothly over time is essential since it is unrealistic to assume that the stochastic future of a time series is stable in the long run. Furthermore, in numerous cases, we observe nonstationary behavior in time series data even within shorter periods.

Economic variables inherently have different degrees of persistence, which can be reconciled with agents' preferences, which differ according to their horizons of interest. 
Economic theory suggests that the marginal utility of agents' preferences depends on the cyclical components of consumption \citep{giglio2015very,bandi2017business}. Furthermore, the frequency-specific preferences of investors \citep{dew2013asset,neuhierl2021frequency,bandi2021} have been documented in the literature, and these preferences have been related to investment horizons according to their risk attitudes \citep{dew2013asset}. 
Such behavior can be observed, for example, under myopic loss aversion, where an agent's decision depends on the valuation horizon. Unexpected shocks or news can alter such preferences, which may result in the generation of transitory and persistent fluctuations of different magnitudes.\footnote{For example, a shock that affects longer horizons may reflect permanent changes in expectations about future price movements. Such a shock may lead to a permanent change in a firm's future dividend payments \citep{balke2002low}. Conversely, a shock that affects shorter horizons may indicate temporary changes in future price movements. For example, suppose that a shock is simply a change in an upcoming dividend payment. This would likely lead to a very short-term change, reflecting the transitory nature of the news.} Importantly, few economic relationships remain constant over decades, years or even months, and the evolution of the economy with unprecedented declines in economic activity results in the generation of very different persistent structures.

To identify time-varying transitory and persistent components of time series data, we propose a time-varying extended Wold decomposition that can be applied to localized heterogeneous persistent structures. We assume that a small neighborhood around a fixed time point remains stationary and allow the coefficients to vary over time according to locally stationary processes \citep{dahlhaus2009}. 
We relax the stationarity assumption of \cite{bandi2019,ortu2020}, who proposed an extended Wold decomposition for stationary processes, and formalize the idea that a time series can be represented by time-varying persistent structures. By preserving the properties of the extended Wold decomposition at each localized point in time, our extension allows various types of locally stationary processes to be decomposed, as the localized components remain uncorrelated.

Our decomposition can be used to identify the duration of the fluctuation that is most relevant to the variability of the time series at a given point in time. Moreover, we elucidate potential drivers of predictability pockets, so our results are largely relevant to the forecasting literature. As noted in \citep{stock2017}, advancements in the forecasting accuracy of economic time series have been limited over the past few decades. First, information is hidden owing to noise and is unevenly distributed over different horizons. Second, economic time series are dynamic and very often nonstationary when modelled over long periods. The approach proposed in this paper can accurately extract the relevant information, enabling the construction of more accurate forecasting models. To our knowledge, we are the first to study the time-varying degree of persistence in time series and identify the pockets of predictability that occur owing to localized heterogeneous persistent structures in time series data, which may represent a new direction for forecasting.\footnote{\cite{lima2007shocks} Note that local or temporary persistent structures exist in many time series.}

The identification of time-varying persistent structures has several advantages over traditional methods based on Wold decomposition. Traditional Wold decomposition, which underlies the vast majority of contemporaneous models, provides aggregate information about the speed, horizon and intensity of persistent shocks. However, these represent coarse descriptions that are insufficient to precisely identify persistent structures in a given period. To capture the heterogeneity of persistent structures, the duration (propagation) of shocks at different levels of persistence and at different points in time must be considered.

To illustrate the broad applicability of our method, we consider two very different datasets that are important to economists: inflation and stock volatility. These applications serve two main purposes. First, we identify a strong persistent structure in the data, which is both highly heterogeneous and time-varying. Second, we illustrate how our model can be used for modeling and forecasting. Both datasets exhibit typical persistence features and are crucial to the field. 

While the properties of aggregate inflation are ultimately of interest to policymakers, the characteristics and determinants of the behavioral mechanisms underlying price setting are important factors that influence inflation over time. The persistence of inflation has direct implications for monetary policy. Similarly, stock market volatility represents a key measure of risk and uncertainty. We show that even in periods of very high persistence, we can identify less persistent subperiods in which transient shocks prevail. Our model, which can accurately identify such dynamics within time-varying persistent structures, is useful for identifying the dynamics driving the data, leading to improved forecasts.


\section{Time Variation of Time Series Components with Different Levels of Persistence}
\label{sec:theory}

The most fundamental justification for time series analysis is Wold's decomposition theorem. According to this theorem, any covariance stationary time series can be represented as a linear combination of its own past shocks \citep{wold1938, hamilton2020time}. This is a critical theorem in the economic literature that is useful to macroeconomists studying impulse response functions and central to tracing the mechanisms underlying economic shocks to improve policy analysis.

However, this is only one possible representation of a time series and is particularly suitable for cases in which we can assume the stationarity of the model. If we cannot assume that the stochastic properties of the data are stable over time, the stationarity assumption may be misleading. Importantly, other representations may capture deeper properties of the time series, which may also change smoothly over time. As argued in the Introduction, we want to explore the properties of time variation and the properties related to different levels of persistence in the time series.

The latter analysis can be performed via persistence-based Wold decomposition \citep{bandi2019,ortu2020}, which enables the decomposition of (multivariate) stationary time series into the sum of orthogonal components associated with their own levels of persistence. These individual components have Wold representations that are defined with respect to scale-specific shocks with heterogeneous persistence. Here, we aim to provide a persistence-based representation for locally stationary processes \citep{dahlhaus1996} and discuss how to decompose locally stationary processes into independent components with different degrees of persistence. With the proposed model, we can study the time variation of components with different degrees of persistence.

\subsection{Locally Stationary Processes}

While stationarity has played an important role in time series analysis for decades because of the availability of natural linear Gaussian modeling frameworks, many economic relationships are not stationary over longer periods. For example, the state of the economy and agent behaviors are often highly dynamic, and the assumption of time-invariant mechanisms generating the data is often unrealistic.\footnote{An exception is nonstationary models in which persistence is generated via integrated or cointegrated processes.} 
A more general nonstationary process may be one that is close to a stationary process at time points in local regions but whose properties (covariances, parameters, etc.) gradually change in a nonspecific way over time. 
The idea that such processes may be stationary only for limited periods and that such processes can still be estimated is not new. These so-called locally stationary processes were introduced in \cite{dahlhaus1996}.

More formally, assume that an economic variable of interest follows a nonstationary process $x_t$ depending on a time-varying parameter model. In this framework, we replace $x_t$ with a triangular array of observations $(x_{t,T};t=1,\ldots,T)$, where $T$ is the sample size, and we assume that we observe $x_{t,T}$ at times $t=1,\ldots,T$. Such a nonstationary process $x_{t,T}$ can be locally approximated \citep{dahlhaus1996} around each rescaled and fixed time $u \approx t/T$ such that $u\in[0,1]$ for a stationary process $x_t(u)$. In other words, under some suitable regularity conditions, $\left|x_{t,T}-x_t(u)\right| = \mathcal{O}_p \left( |\frac{t}{T}-u|+\frac{1}{T} \right)$. While stationary approximations vary smoothly over time as $u \mapsto x_t(u)$, locally stationary processes can be interpreted as processes whose (approximately) stationary properties change smoothly over time. The main properties of $x_{t,T}$ are thus encoded in the stationary approximation; hence, in the estimation process, we focus on the quantities $\mathbb{E}\left[ g(x_t(u),x_{t-1}(u),\ldots) \right]$ with some function $g(. )$ as natural approximations of $\mathbb{E}\left[ g(x_{t,T},x_{t-1,T},\ldots) \right]$.

Crucially, a linear locally stationary process $x_{t,T}$ can be represented by a time-varying MA($\infty$):
\begin{equation}
\label{eq:wold0}
x_{t,T} = \sum_{h=-\infty}^{+\infty} \alpha_{t,T} \left(h\right) \epsilon_{t-h},
\end{equation}
where the coefficients $\alpha_{t,T} \left(h\right) $ can be approximated under certain (smoothness) assumptions (see Assumption \ref{propLSP} in Appendix \ref{sec:app}) with coefficient functions $\alpha_{t,T} \left(h\right) \approx \alpha\left(t/T,h\right)$. In addition, $\epsilon_t$ are independent random variables with $\E\epsilon_t = 0$, $\E\epsilon_s\epsilon_t=0$ for $s\ne t$, and $\E|\epsilon_t| < \infty$. The construction with $\alpha_{t,T} \left(h\right)$ and $\alpha\left(t/T,h\right)$ appears complicated, but a function $\alpha\left(u,h\right)$ is needed for rescaling and to impose the smoothness condition, whereas the additional use of $\alpha_{t,T} \left(h\right)$ makes the class rich enough to cover autoregressive models (see Theorem 2.3. in \cite{dahlhaus1996}), as discussed later.

It is straightforward to construct a stationary approximation (with existing derivative processes) as follows:
\begin{equation}
\label{eq:wold}
x_t(u) = \sum_{h=-\infty}^{+\infty} \alpha\left(u,h\right) \epsilon_{t-h},
\end{equation}
where at each fixed time $u$, the original process $x_{t,T}$ can be represented as a linear combination of uncorrelated components with time-varying impulse response functions $\alpha\left(u,h\right)$. 
Note that thus far, the process is assumed to have a zero mean $\mu(t/T)$. While this may be unrealistic in numerous applications, we return to this assumption later in the estimation.

\subsection{Time-Varying Extended Wold Decomposition}

Having a representation that allows for time variation in the impulse response function, we further introduce a localized persistent structure. We use the extended Wold decomposition (EWD) of \cite{bandi2019,ortu2020}, which allows us to decompose the time series into several components with different levels of persistence. The decomposition has considerable advantages in helping us understand the persistence-related dynamics of economic time series data and improves forecasting performance, as many economic time series exhibit heterogeneous persistent structures (across horizons and scales).

Importantly, we argue that in addition to recovering the heterogeneous persistent structures of typical economic time series, we need to localize such structures. Localization and persistence decomposition can dramatically increase our understanding of dynamic economic behavior by allowing persistent structures to change smoothly over time. In turn, models built with such an understanding may have substantially improved forecasting performance, as discussed later with empirical examples.

In particular, we propose a model in which locally stationary processes are used to capture the dynamics of heterogeneous persistence. As we can express locally stationary processes using the time varying MA($\infty$) representation, we can adapt the EWD under alternative assumptions and localize the decomposition results. By inheriting the properties of the EWD at each localized time point, our extension provides decompositions of any locally stationary processes according to the assumption \ref{propLSP} in Appendix \ref{sec:app}, as the localized components remain uncorrelated.

The proposition \ref{prop} formalizes the main result and presents the time-varying EWD model.

\begin{prop}[Time-Varying Extended Wold Decomposition (TV-EWD)]
\label{prop}
If $x_{t,T}$ is a zero-mean, locally stationary process in the sense of Assumption \ref{propLSP} in Appendix \ref{sec:app} that has a representation $x_{t,T} = \sum_{h=-\infty}^{+\infty} \alpha_{t,T}(h) \epsilon_{t-h}$, then it can be decomposed as
\begin{equation}
x_{t,T}=\sum_{j=1}^{+\infty} \sum_{k=0}^{+\infty} \beta_{t,T}^{\{j\}}(k) \epsilon_{t-k2^j}^{\{j\}},
\end{equation}
where for any $j \in \mathbb{N}, k \in \mathbb{N}$
\begin{equation}
\beta_{t,T}^{\{j\}}(k)= \frac{1}{\sqrt{2^j}}  \left[ \sum_{i=0}^{2^{j-1}-1} \alpha_{t,T}(k2^j+i) - \sum_{i=0}^{2^{j-1}-1} \alpha_{t,T}\left(k2^j+2^{j-1}+i\right) \right],
\end{equation}
\begin{equation}
\epsilon_t^{\{j\}} = \frac{1}{\sqrt{2^j}}  \left( \sum_{i=0}^{2^{j-1}-1} \epsilon_{t-i} - \sum_{i=0}^{2^{j-1}-1} \epsilon_{t-2^{j-1}-i} \right),
\end{equation}
where the coefficients $\beta_{t,T}^{\{j\}}(k)$ can be approximated according to Assumption \ref{propLSP_beta} in Appendix \ref{sec:app} with coefficient functions $\beta_{t,T}^{\{j\}}(k) \approx \beta^{\{j\}}\left(t/T,k\right)$. Here, $\epsilon_t$ are independent random variables with $\E\epsilon_t = 0$, $\E\epsilon_s\epsilon_t=0$ for $s\ne t$, $\E|\epsilon_t| < \infty$, and $\sum_{k=0}^{\infty}\left(\beta_{t,T}^{\{j\}}(k)\right)^2 < \infty$ for all $j$.
\end{prop}

\begin{proof}
This follows directly from the properties of locally stationary processes \citep{dahlhaus1996} and the EWD \citep{bandi2019,ortu2020}.
\end{proof}

Proposition \ref{prop} formalizes the discussion about the representation of a time series that provides a decomposition into $j$ uncorrelated persistent components that can vary smoothly over time. In particular, it allows the construction of a stationary approximation (with existing derivative processes) for the process $x_{t,T}$ with time-varying uncorrelated persistent components.
\begin{equation}
x_t^{\{j\}}(u)=\sum_{k=0}^{+\infty} \beta^{\{j\}}(u,k) \epsilon_{t-k2^j}^{\{j\}}
\end{equation}
Furthermore, the original process can be reconstructed as
\begin{equation}
x_t(u)=\sum_{j=1}^{+\infty} x_t^{\{j\}}(u),
\end{equation}
In other words, we can decompose the time series into uncorrelated components with different levels of persistence at any fixed point in time. Note that $\epsilon_t^{\{j\}}$ is a localized MA($2^j-1$) with respect to fundamental components of $x_{t,T}$, 
and $\beta^{\{j\}}(u,k)$ is the time-varying multiscale impulse response associated with scale $j$ and time shift $k 2^j$ at a fixed time approximated by $u$.

This decomposition allows us to examine time-varying impulse responses at different levels of persistence. A scale-specific impulse response provides precise information on how a unit shock to the system propagates at different horizons at a given point in time. For example, in the case of daily data, the first scale, $j=1$, describes how a unit shock dissipates over two days, the second scale, $j=2$, describes this process over four days, and so on.

\subsection{Obtaining Time-Varying Persistent Structures from Data}

The proposed approach can be used to identify localized time-varying persistent structures in time series data. Next, we use the approach to build a parametric model that can be used to increase forecasting accuracy. The first step is to obtain the quantities described in the previous section.

Considering the model assumptions, we assume that the economic variable of interest follows a time-varying parameter autoregressive (TV-AR($p$)) model with $p$ lags.
\begin{equation}
x_{t,T}=\phi_0\left(\frac{t}{T}\right)+\phi_1\left(\frac{t}{T}\right)x_{t-1,T}+\ldots+\phi_p \left(\frac{t}{T}\right)x_{t-p,T} + \epsilon_t
\label{eq:tvar}
\end{equation}
This model has a representation according to in Proposition \ref{prop} and can be locally approximated, under appropriate conditions, by a stationary process $x_{t,T} \approx x(u)$ for a given $t/T \approx u$ with $\phi_i\left(\frac{t}{T}\right) \approx \phi_i(u)$. To obtain the decomposition, we need to identify the time-varying coefficient estimates $\widehat{\Phi}\left(\frac{t}{T}\right)=\left(\widehat{\phi}_1\left(\frac{t}{T}\right), \ldots,\widehat{\phi}_p\left(\frac{t}{T}\right)\right)'$ based on the centered data $\widetilde{x}_{t,T} = x_{t,T}-\widehat{\phi}_0\left(\frac{t}{T}\right)$. This is particularly important for some datasets exhibiting clear temporal trends, whereas this step may be negligible for other datasets. 

\subsubsection{Local Linear Estimation}

We estimate the coefficient functions $\phi_i\left(\frac{t}{T}\right)$ using the local linear method. Local linear estimation methods have been used in nonparametric regression estimation because of their attractive properties, such as high efficiency, reduced bias and adjustments for boundary effects \citep{fan1996local}. Assuming that each $\phi_i\left(\frac{t}{T}\right)$ has a continuous second-order derivative in the interval $[0,1]$, the function can be approximated around $u$ by a linear function via the first-order Taylor expansion.
\begin{equation}
\phi_i\left(\frac{t}{T}\right) \approx \phi_i(u) + \phi'_i(u)\left(\frac{t}{T}-u\right),
\end{equation}
where $\phi'_i(u)=\partial \phi_i(u)/\partial u$ is the first derivative. On the basis of the local approximation of the model \ref{eq:tvar}, we can minimize the locally weighted sum of squares to estimate the parameters $\Phi\left(u\right)=\left\{\phi_1(u),\ldots,\phi_p(u)\right\}'$.
\begin{equation}
\left\{\widehat{\Phi}\left(u\right),\widehat{\Phi}'\left(u\right)\right\} = \underset{(\theta,\theta')\in \mathbb{R}^2}{\mathrm{argmin}} \sum^T_{t=1} \left[\widetilde{x}_{t,T} - U_{t,T}^{\top}\theta-\left(\frac{t}{T}-u \right)U_{t,T}^{\top}\theta' \right]^2 K_b\left(\frac{t}{T}-u\right) 
\label{eq:tvols},
\end{equation}
where $U_{t,T}=\left(x_{t-1,T},x_{t-2,T},\ldots,x_{t-p,T} \right)^{\top}$, and $K_b(z) = 1/b K(z/b)$ is a kernel function where $b=b_T>0$ is a bandwidth satisfying the conditions that $b\rightarrow 0$ and $T b \rightarrow \infty$ as $T \rightarrow \infty$. Note that $b$ controls the amount of smoothing used in the local linear estimation. Roughly, we fit a set of weighted local regressions with an optimal window size chosen according to bandwidth $b$, as discussed below. The estimator has a general expression that can be obtained via elementary computations \citep{fan1996local}, and the estimates of the coefficients are asymptotically normally distributed under some regularity conditions.

Note that we use centered data $\widetilde{x}_{t,T} = x_{t,T}-\widehat{\phi}_0\left(\frac{t}{T}\right)$, which we obtain as
\begin{equation}
\left\{\widehat{\phi}_0\left(u\right),\widehat{\phi}'_0\left(u\right)\right\} = \underset{(\mu, \mu')\in \mathbb{R}^2}{\mathrm{argmin}} \sum^T_{t=1} \left[x_{t,T} - \mu - \mu' \left(\frac{t}{T}-u \right)\right]^2 K_b\left(\frac{t}{T}-u\right),
\label{eq:tvols_const}
\end{equation}

The local linear estimator is sensitive to the choice of bandwidth $b$; therefore, choosing an appropriate bandwidth in applications is critical. Here, we follow the commonly used cross-validation bandwidth for the time series case \citep{hardle1992kernel}

Finally, after obtaining the time-varying coefficients, we express the time series as a (local) Wold's MA representation with the process $\epsilon_t$ 
(see Eq. \ref{eq:wold} and Appendix \ref{sec:app_wold} for a detailed implementation of the recursion). We then use the result of Proposition \ref{prop} to obtain the horizon-specific impulse response coefficients associated with scale (horizon) $j$ and time shift $k2^j$. The decomposition must be truncated at a finite number of scales $J$ and observations $T$. Note that the choice of window $T$ and bandwidth limits the maximum level of decomposition $J$, and the smaller the window is, the lower the level of decomposition possible. Thus, these choices limit the detection of components with higher degrees of persistence. Therefore, a finite version of the time-varying EWD at a given time $u$ is considered.
\begin{equation}
\widehat{x}_{t,T}=\sum_{j=1}^{J} \widehat{x}_{t,T}^{\{j\}} + \pi_t^{\{J\}}=\sum_{j=1}^{J} \sum_{k=0}^{N-1} \widehat{\beta}^{\{j\}}(u,k) \widehat{\epsilon}_{t-k2^j}^{\{j\}} + \pi_t^{\{J\}}(u),
\end{equation}
where $\widehat{\epsilon}_t^{\{j\}} = \frac{1}{\sqrt{2^j}}  \left( \sum_{i=0}^{2^{j-1}-1} \widehat{\epsilon}_{t-i} - \sum_{i=0}^{2^{j-1}-1} \widehat{\epsilon}_{t-2^{j-1}-i} \right)$. In addition, $\pi_t^{\{J\}}$ is a residual component at scale $J$, defined as $\pi_t^{\{J\}}(u) = \sum_{k=0}^{+\infty} \gamma_k^{\{J\}}(u) \epsilon_{t,T}^{\{J\}}$, and $\epsilon_{t}^{\{J\}}=\frac{1}{\sqrt{2^J}}\sum_{j=0}^{2^J-1} \epsilon_{t-i}, \gamma_k^{\{J\}}(u)=\frac{1}{\sqrt{2^J}}\sum_{j=0}^{2^J-1} \alpha(u,k2^J+i)$. The estimates of scale-specific coefficients $\widehat{\beta}^{\{j\}}(u,k)$ are computed as $\widehat{\beta}^{\{j\}}(u,k)= \frac{1}{\sqrt{2^j}}  \left( \sum_{i=0}^{2^{j-1}-1} \widehat{\alpha}(u,k2^j+i) - \sum_{i=0}^{2^{j-1}-1} \widehat{\alpha}(u,k2^j+2^{j-1}+i) \right)$. The residual component $\pi_t^{\{J\}}(u)$ is usually negligible, so we do not consider it in the estimation.

\subsection{Forecasting Models with Time-Varying Persistence}
\label{sec:forecasting}

One of the main advantages of our model is that it captures a persistently changing structure that can be used for forecasting purposes. In many cases, it is unrealistic to assume that the stochastic structure of a time series remains stable over longer periods. Furthermore, nonstationarity can be observed in shorter time series, and forecasting under the assumption of stationarity can be misleading. A common approach for dealing with nonstationarity is to assume a model with a smoothly changing trend and variance but with a stationary error process \citep{stuaricua2005nonstationarities}. While several authors have considered forecasting in a locally stationary setting \citep{dette2022prediction}, our approach extends these models by exploring the smoothly changing persistent structures within the data.

The main aim of a forecasting model based on time-varying persistence decomposition is to provide an $h$-step-ahead predictor for the unobserved $x_{T+h,T}$ from the observed $x_{1,T},\ldots,x_{T,T}$ data. Given the values of the parameters $\widehat{\beta}^{\{j\}}(u,k)$ and $\widehat{\epsilon}_t^{\{j\}}$, we can write the original series $x_{t, T}$ as a sum of a deterministic time trend $\widehat{\phi}_0\left(t/T\right)$ and the orthogonal persistence components, $\widehat{x}_{t,T}^{\{j\}}$, as follows:
\begin{equation}
x_{t,T}=\widehat{\phi}_0\left(t/T\right) + \sum_{j=1}^{J} \widehat{x}_{t,T}^{\{j\}}.
\end{equation}
Conditional $h$-step-ahead forecasts can then be obtained directly by combining the trend forecast with the forecast of the scale components. Additionally, we can take advantage of including only a selection of the most informative persistence components when building the model. In other words, we can select components from a set $\mathbb{D}$ containing available persistence components $\{1,\ldots,J\}$ such that they provide the greatest increase in accuracy with respect to a specific criterion, such as the out-of-sample loss function. Furthermore, the introduction of weights $w^{\{j\}}$, which can be used to identify the importance of a particular $j$ in the forecasted time series, can improve the forecasting performance. This enables us to emphasize or suppress certain horizons in the forecasting model. Thus, the general forecasting model can be written as follows:
\begin{equation}
\E_t[x_{T+h,T}]=\E_t[\widehat{x}_{T+h,T}^{\{0\}}] + \sum_{j\in \mathbb{D}} \widehat{w}^{\{j\}}\E_t[x_{T+h,T}^{\{j\}}].
\end{equation}
In our work, we start with all persistence components on scales $j=1,\ldots,J$, where $J$ is chosen arbitrarily provided that $2^J<T$. For practical purposes, we recommend selecting a smaller value of $J$, since higher scales result in coarser time localization and lower accuracy. This is because the number of available time points decreases as $J$ increases. Therefore, for implementation purposes, the connection is important when choosing window sizes for local estimation. We then estimate the values of $\widehat{w}^{\{j\}}$ using linear regression. Note that if we restrict all $w^{\{j\}}=1$, we strictly follow the EWD. In addition, we can optimize the weights with respect to the out-of-sample metric and use shrinkage techniques to select the optimal components in $\mathbb{D}$ and the associated weights.

The conditional expected value of the trend, $\E_t[\widehat{x}_{T+h,T}^{\{0\}}]$, is forecasted as TV-AR(1)\footnote{The process is forecasted with the local linear estimator in \ref{eq:tvols}, with the Epanechnikov kernel having a width denoted as the kernel width - 2 in the subsequent forecasting exercises.}, and the conditional expectation of the scale components $\E_t[x_{t+h,T}^{\{j\}}]$ is computed by shifting the localized Wold coefficients by the $h$ steps. Thus,
\begin{equation}
\E_t[x_{t+h,T}] = \sum_{p=-\infty}^{+\infty} \alpha_{t,T} \left(p+h\right) \epsilon_{t-p}.
\end{equation}
This shift is then processed via the time-varying EWD as
\begin{equation}
\E_t[x_{t+h,T}] = \sum_{j=1}^{J} \sum_{k=0}^{N-1} \widehat{\beta}^{\{j\}}(u,k,h) \widehat{\epsilon}_{t-k2^j}^{\{j\}},
\end{equation}
where
\begin{equation}
\widehat{\beta}^{\{j\}}(u,k,h)= \frac{1}{\sqrt{2^j}}  \left( \sum_{i=0}^{2^{j-1}-1} \widehat{\alpha}(u,k2^j+i+h) - \sum_{i=0}^{2^{j-1}-1} \widehat{\alpha}(u,k2^j+2^{j-1}+i+h) \right).
\end{equation}

Therefore, once the time-varying EWD of $x_{t,T}$ is known, an $h$-step-ahead forecast can be obtained from the above definitions. The model is highly adaptable to different contexts, but its effective application depends on tailoring the model to the data, which requires making several choices related to the sample size and the richness of the data structure being studied. We provide a package for this model called \texttt{tvPersistence.jl} in \textsf{JULIA} at \url{https://github.com/barunik/tvPersistence.jl}.\\

\subsection{Pockets of Predictability and Benchmark Prediction Models}

A natural benchmark model to consider in addition to the random walk (RW) model commonly used in the inflation forecasting literature is a simple AR model (AR($p$)) that captures persistence within covariance stationary data with $p$ parameters, $x_t=\phi_0+\phi_1x_{t-1},\ldots,\phi_px_{t-p}+\epsilon_t$, where $\epsilon_t$ is a \textit{iid} random variable. The EWD of \cite{ortu2020} can be used as a natural benchmark model to capture the persistence of stationary data with heterogeneous components: 
\begin{equation}
x_{t}=\sum_{j=1}^{J} \sum_{k=0}^{N-1} \beta^{\{j\}}_k \epsilon_{t-k2^j}^{\{j\}},
\end{equation}
where $\beta^{\{j\}}_k$ are scale-specific moving average coefficients associated with the approximated $AR($p$)$ model. Then, a time-varying autoregressive TV-AR($p$) model in which the coefficients are allowed to vary over time can be formulated as
\begin{equation}
x_{t,T}=\phi_0\left(\frac{t}{T}\right)+\phi_1\left(\frac{t}{T}\right)x_{t-1,T}+\ldots+\phi_p \left(\frac{t}{T}\right)x_{t-p,T} + \epsilon_t.
\label{eq:tvar_model}
\end{equation}
This model can be estimated via local linear regression, as outlined above.

To evaluate the benefits of using localized heterogeneously persistent components for predictions, we follow the ``pocket of predictability'' approach proposed by \cite{farmer2023pockets}, who suggested using a measure of relative prediction accuracy, defined as the squared error difference (SED) between a benchmark forecast $\overline{x}_{t+h}$ and the forecast produced by the local regression model $\widehat{x}_{t+h,T}$.
\begin{equation}
\text{SED}^h_t=\left(x_t-\overline{x}_{t+h}\right)^2-\left(x_t-\widehat{x}_{t+h,T}\right)^2
\end{equation}
Periods with $\text{SED}_t^h>0$ indicate that the local regression model produces a more accurate forecast than the benchmark model according to the squared error. To identify the periods with more accurate forecasts, which \cite{farmer2023pockets} called ``pockets of predictability'', we estimate
\begin{equation}
\widehat{\text{SED}}^h_t=\widehat{\gamma}_0\left(\frac{t}{T}\right)+\widehat{\gamma}_1\left(\frac{t}{T}\right)\nu_t
\end{equation}
using a one-sided Epanechnikov kernel and identify pockets of predictability during which $\widehat{\text{SED}}^h_t>0$. To measure the total amount of predictability gained by exploring the pockets, we can define
\begin{eqnarray}
1/T \sum_{\tau=1}^T\mathcal{LR}(\tau) \mathbf{1}_{\{ \widehat{\text{SED}}^h_{\tau}>0\}}
\end{eqnarray}
where we use $\mathcal{LR}(u)=\frac{ \sum_{t=1}^T K_b\left(t/T-u\right)\mathcal{L}^{\text{OOS}}_{\widehat{x}_{t+h,T}}}{\sum_{t=1}^T K_b\left(t/T-u\right)\mathcal{L}^{\text{OOS}}_{\overline{x}_{t+h,T}}}$ with $\mathcal{L}^{\text{OOS}}_{(.)}$ as the root mean square error (RMSE) or mean absolute error (MAE) loss function of a given model. The ratio $\mathcal{LR}(u)$ is a ratio of the out-of-sample local RMSE (or MAE) of the localized prediction model to that of the benchmark and quantifies the gains associated with the use of the localized predictability pockets.

The definition builds on the practice, started by \cite{welch2008comprehensive}, of studying the evolution of predictability over time according to the sums of the squared forecast error differences. These differences are also the basis for formal comparisons of economic forecasting performance \citep{farmer2023pockets}. However, $\text{SED}^h_t$ allows temporary predictability to be identified via local estimates of trends in the relative forecast accuracy. Note that we do not need to calculate the standard errors for the estimates to identify these pockets. This property is particularly important for out-of-sample forecasts. One could also easily specify that a pocket can be triggered only after a certain number of periods.

We also use simulations to determine the extent to which our approach spuriously identifies pockets of predictability. Since we repeatedly compute local, overlapping test statistics, the question is whether we find more pockets than would be expected by chance, given a reasonable persistence model for the underlying dynamics. To address the spuriously generated pockets, we simulate highly persistent variables to match the persistence found in the data using AR models and use bootstrapping with replacement to construct bootstrap samples, based on which we compute the bootstrapped confidence intervals for the $\widehat{\text{SED}}^h_{\tau}$ statistics. This simple strategy allows us to discriminate between spurious and non-spurious pockets by examining the statistics for each pocket and calculating the percentage of simulations with at least one pocket matching that value. We use pockets that have less than a 5\% chance of being randomly generated.


\section{Monte Carlo Study: Time-Varying Persistence and TV-EWD}
\label{sec:simulations}

Before moving on to empirical applications, we conduct a brief Monte Carlo exercise to understand the advantages of our method and the nature of the dependent structures captured by our model. We consider simulated data with different sources of persistence and test the performance of the TV-EWD approach against alternative methods. These sources of persistence are chosen to mimic different types of persistence found in the literature and considered in our applications for forecasting inflation and volatility in later sections, such as simple persistence generated by an AR model, structural breaks, or more complicated structures generated by time-varying processes. To motivate the cases where the focus on time-varying persistent structures is most helpful, we consider the following four data generation processes (DGPs). First, we generate the data using a simple AR(1) process:
\begin{equation}
\text{DGP I:} \hspace{1cm} x_t=\phi x_{t-1}+\epsilon_t
\end{equation}
with three levels of persistence $\phi \in \{0.5,0.75,0.9\}$. Second, we consider a process in which the persistence varies smoothly over time at different rates and compare the performance of our method with that of alternative benchmarks in which the process becomes increasingly ``locally stationary'', formalized as:
\begin{eqnarray}
 \nonumber x_t&=&\phi_t x_{t-1}+\epsilon_t \\
\text{DGP II :} \hspace{1cm}\phi_t &=& 0.95 \sin(2\pi kt/T)
\end{eqnarray}
With this approach, for $k \in \{0.75,1.5,3\}$, we have three cases with persistent structures that change faster within the interval $(-0.95,0.95)$ using the $\sin(.)$ function, which completes 0.75, 1.5 or 3 cycles within the period. In other words, for the fastest process with $k=3$, the persistence parameter is changed from 0.95 to -0.95 and back three times.

Third, we use two components that drive time-varying persistence, a transitory component and a more persistent component, as formulated in DGP III:
\begin{eqnarray}
 \nonumber x_t&=&\phi_t x_{t-1}+\epsilon_t \\
\text{DGP III :} \hspace{1cm} \phi_t &=& 0.95 \sin(2\pi 3 t/T) + 0.95 \sin(2\pi 1.5t/T).
\end{eqnarray}
Finally, we consider a process with structural breaks that contain abrupt changes in parameters at different times, reflecting shifts in the process of data generation. This process involves structural breaks that occur at times $T_1, T_2, \dots, T_m$ as follows:
\begin{equation}
\text{DGP IV:} \hspace{1cm}  x_t = \sum_{j=0}^{m} \left(\phi_j x_{t-1} \right) \cdot \mathbf{1}_{\{ T_{j} < t \le T_{j+1} \}} + \varepsilon_t,
\end{equation}
where $\phi_j$ is the persistence coefficient in regime $j=0,1,2,3$ and where $\mathbf{1}_{t<T}$ is an indicator function. Specifically, we consider three cases: a case with one regime $m=1$, such as $\phi_j=\{-0.8,0.8\}$ and $T_1=800$; a case with two regimes $m=2$, such as $\phi_j=\{-0.8,0.8,-0.8\}$ and $T_1=600$; $T_2=950$; and a case with three regimes $m=3$, such as $\phi_j=\{0.8,-0.8,0.8,-0.8\}$ and $T_1=600$; $T_2=800$; and $T_3=800$.
\begin{table}[h!] 
\caption{\footnotesize{Medians of the mean square errors (MSEs) of AR(1), AR(3), EWD, TV-AR(3) and TV-EWD relative to the random walk model from 200 realizations of the following DGPs. DGP I: $x_t=\phi x_{t-1}+\epsilon_t$, DGP II: $x_t=\phi_t x_{t-1}+\epsilon_t$ with $\phi_t = 0.95 \sin(2\pi kt/T)$ and $k \in \{0.75,1.5,3\}$, DGP III: with $\phi_t = 0.95 \sin(2\pi 3 t/T) + 0.95 \sin(2\pi 1.5t/T)$ and DGP IV: $x_t = \sum_{j=0}^{m} \left(\phi_j x_{t-1} \right) \cdot \mathbf{1}_{\{ T_{j} < t \le T_{j+1} \}} + \varepsilon_t$ with $m=3$ regimes $j=0,1,2,3$.}}
\begin{center} 
\footnotesize 
\begin{tabular}{lccccccccccccc}
  \toprule
  & \multicolumn{3}{c}{DGP I ($\Phi$)} & & \multicolumn{3}{c}{DGP II ($k$)} & & DGP III & & \multicolumn{3}{c}{DGP IV ($m$)}\\ 
\cmidrule(r){2-4} \cmidrule(r){6-8} \cmidrule(r){10-10} \cmidrule(r){12-14}
  & 0.5 & 0.75 & 0.9 && 0.75 & 1.5 & 3 & & & & 1 & 2 & 3 \\
  \cmidrule(r){2-4} \cmidrule(r){6-8} \cmidrule(r){10-10} \cmidrule(r){12-14}
  AR(1) &0.868 & 0.938 & 0.977 & & 1.233 & 1.105 & 0.680 & & 0.765 & & 0.986 & 0.834 & 0.662 \\
  AR(3) & 0.869 & 0.939 & 0.979 & & 1.031 & 0.822 & 0.508 && 0.721 & & 0.757 & 0.619 & 0.504 \\
  EWD & 0.887 & 1.026 & 1.458 & & 1.143 & 0.855 & 0.519 & & 0.734 && 0.806 & 0.637 & 0.512 \\
  TV-AR(3) & 1.076 & 1.385 & 1.581 & & 1.047 & 0.819 & 0.532 && 0.727 & & 0.830 & 0.641 & 0.554 \\
  TV-EWD & 0.900 & 1.013 & 1.254 & & 0.884 & 0.668 & 0.452 && 0.638 & & 0.628 & 0.490 & 0.438 \\
  \bottomrule
\end{tabular}
\end{center} 
\label{tab:sims} 
\end{table}

To compare the performance of the TV-EWD approach with that of benchmark models, we consider models that capture persistence, such as AR(1) and AR(3). In addition, we consider EWD, TV-AR(3) and TV-EWD models, which capture persistence due to heterogeneous components and persistence due to time variations in the process parameters. 
The selection of these models is not arbitrary. They are used as benchmarks in the later applications and also used here to understand the sources of the improvements in prediction performance achieved by the TV-EWD method.

Table \ref{tab:sims} shows the medians of the out-of-sample forecast performance metrics of all the models estimated based on the 200 randomly generated DGPs, measured via the mean square errors (MSEs) relative to the RW model. Considering our real-world applications, we consider T=1268 observations, keeping 610 as in-sample forecasts and the rest as out-of-sample forecasts. For the EWD and TV-EWD models, we consider $J=5$ scales.

When the process is only very persistent, as in DGP I, more complex models that capture persistence due to shocks caused by different components with heterogeneous structures (EWD) and time variations (TV-AR(3) and TV-EWD) do not add value to the simple AR(1) model, and we should not expect complicated models to add value. However, with more complex dependence structures, the benefits of such models become clear. Specifically, for DGP II and DGP III, both types of persistent approaches (AR(1) and EWD, capturing unconditional persistence, and TV models capturing localized persistence) lead to performance improvements. 
Most importantly, their combination in the TV-EWD approach provides the best result, which is particularly relevant in situations in which shocks have localized heterogeneous structures. DGP III, which combines coefficients with transitory and persistent time variations, has a rather complicated structure but is well suited for use with the TV-EWD method. Similarly, the TV-EWD method performs well with an increasing number of structural breaks, as in DGP IV.


\section{Time-Varying Persistence in Data}

The proposed approach is useful for any problem in economics and finance that requires working with data whose persistent structures are expected to change over time. Here, we want to demonstrate the importance of identifying persistent structures within two different and important datasets. Both time series are different in nature but share the common feature of smoothly changing persistent structures.

In the first example, we examine the time-varying persistent structures in an inflation time series, which is one of the most important time series in macroeconomics. While the properties of aggregate inflation are ultimately of interest to policymakers, an important factor underlying the behavior of inflation over time is the characteristics and determinants of the behavioral mechanisms underlying price setting. The persistence of inflation has direct implications for monetary policies. While time-varying models have been used in the literature \citep{lansing2009time} to capture time variations within such time series data, several authors have considered the decomposition of inflation into transitory and permanent components \citep{stock2007has}. Here, we build a more flexible model to explore the time-varying persistent structures within inflation time series and identify pockets of predictability due to local structures within the inflation series.

The second example is stock market volatility. Similar to inflation, stock market volatility is one of the key measures indicating risk and uncertainty. The study of heterogeneous persistent structures within such data, which evolve dynamically over time, is useful to a wide audience.

\subsection{Time-Varying Pockets of Persistence in U.S. Inflation}

The data used in this analysis are the Personal Consumption Expenditures (PCE) price index\footnote{The Personal Consumption Expenditures price index measures the prices that U.S. consumers pay for goods and services. The change in the PCE price index captures inflation or deflation across a wide range of consumer expenditures.} available on the Federal Reserve of St. Louis website\footnote{\url{https://fred.stlouisfed.org}} as a proxy for U.S. inflation. Our dataset contains 781 monthly observations over the period from January 1959 to February 2023, and we examined the log change in the index.

Inflation is an interesting time series for our analysis because the shocks that drive inflation have varying degrees of persistence and tend to change over time. Inflation is driven by different shocks in stable and turbulent periods. The smoothly changing persistent structure of inflation remains hidden to the observer when classic time series tools such as impulse response functions are used for analysis.

Figure \ref{fig:PCE_persistence} illustrates this idea using our TV-EWD approach. Specifically, the plot shows the ratio of $\widehat{\beta}^{j}(t/T,k)$ to the sum across scales $\sum_j \widehat{\beta}^{j}(t/T,k)$, where the $j$ scales represent 2-, 4-, 8-, 16-, and 32-month persistent shocks at $k=1$. That is, we investigate the relative importance of information at the $2^j$ horizon in the multiscale impulse response function.
 \begin{figure}[t!]
            \begin{center}
                \includegraphics[scale=0.5]{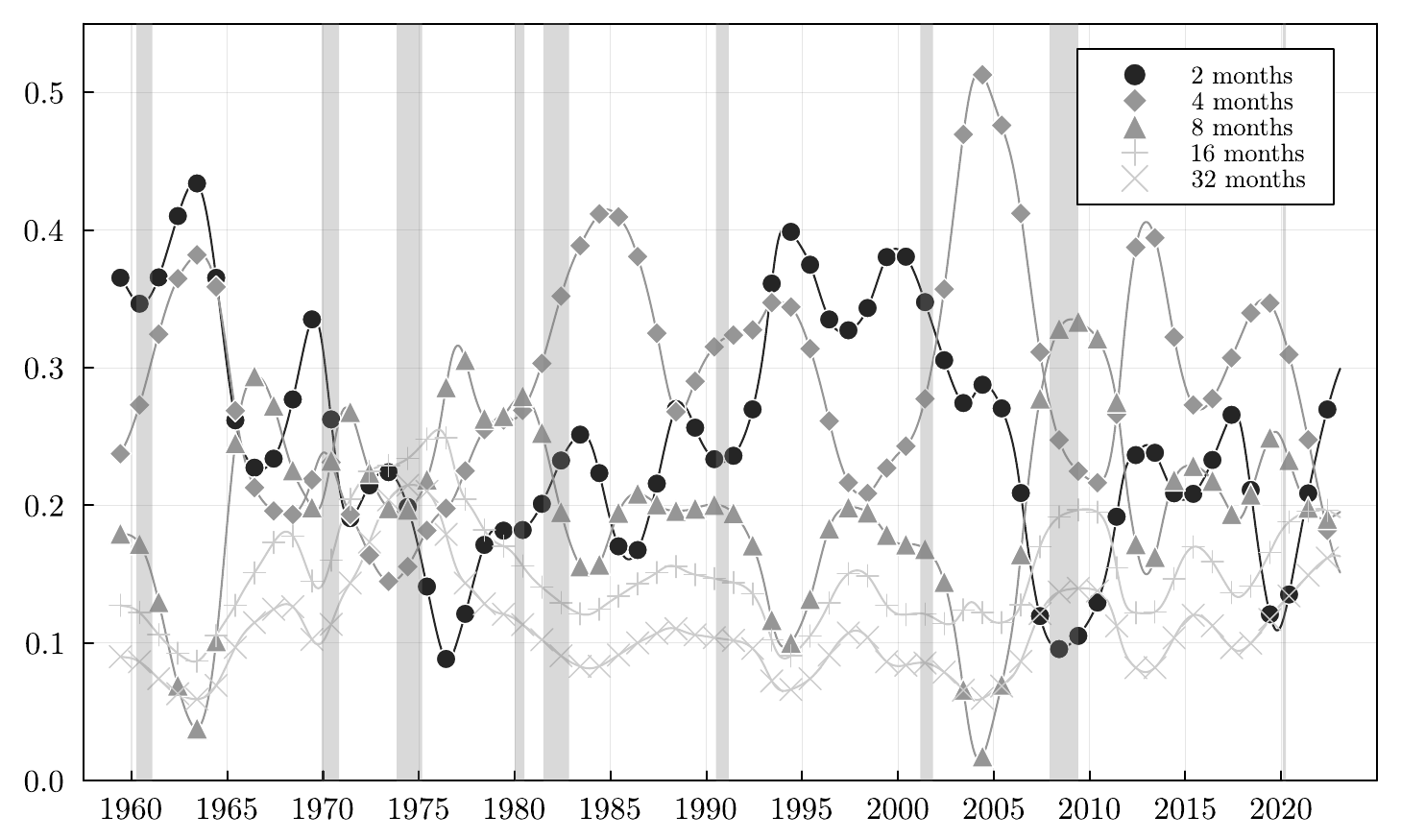}
            \end{center}
            \caption{\footnotesize{Time-varying persistent structure within a U.S. inflation dataset. The plot shows the ratios of $\widehat{\beta}^{j}(t/T,1)/\sum_j \widehat{\beta}^{j}(t/T,1)$ with j corresponding to the 2-, 4-, 8-, 16-, and 32-month persistent shocks over the period from January 1959 to February 2023. The gray areas mark the NBER recession periods.}}
            \label{fig:PCE_persistence}
\end{figure} 
For each period, we identify the persistent shock structures affecting the U.S. inflation series. There are periods where most of the shocks have transitory durations of 2 or 4 months. For example, transitory shocks of up to two months had the largest share of information in the years 1959--1964, 1968--1970, and 1994--2001. In contrast, the changes in the years 1966--1967, 1976--1978 and 2008--2010 were driven mainly by more persistent shocks with durations of up to 8 months. Interestingly, distinct persistent structures are observed during different crises, which are marked by NBER recession periods in the plot. During the recession from November 1973 to March 1975, inflation was driven mainly by shocks lasting 16 and 32 months, which were therefore very persistent.

We find that the persistent structures in the inflation time series are rich and change smoothly over time. We next examined how the identified local persistent structures are useful for forecasting U.S. inflation.

\subsubsection{Forecasting Inflation}

The exploratory analysis in the previous section shows that the persistent structures in the U.S. inflation series are rich and vary smoothly over time. We want to explore this feature to propose a forecasting model based on the precisely identified time-varying persistent structures. Since our approach combines two types of time series representations, one allowing shocks to vary over time and the other allowing shocks to have different degrees of persistence, it is natural to compare the performance of the TV-EWD model with that of the TV-AR and EWD models, which capture time variations and heterogeneous persistent structures separately. As a natural comparison, we also use a simple AR model, which is often used in the literature and aggregates both dimensions. We compare the performance of all these models with that of a RW model, which is commonly used in the literature.

For estimation and forecasting with the TV-EWD model, we use the procedures described in Section \ref{sec:forecasting}, with $J=5$ scales and kernels of 0.6 for the trend and 0.2 for the moving average parameters. We choose two lags in the AR approximation model to minimize out-of-sample losses. Importantly, the TV-AR models underlying our decomposition typically require much lower-order AR coefficients, as the time-varying constant captures a substantial part of the persistence due to various local trends, seasonality, inertia in inflation dynamics and abrupt shifts \citep{stock2007has,lansing2009time,primiceri2005time}. 
We therefore use two lags, which provide the most accurate forecasts from the time-varying parameters used in the other models. 
Unlike time-varying models, the AR$(p)$ model typically requires higher-order lags, reaching up to 12, because a constant inflation process is assumed within the model, and the model struggles to account for nonlinear relationships. The choice of lags is often sensitive, with increasing numbers of lags leading to potential overfitting and poorly identified coefficients. Since our main objective is to identify pockets of predictability that occur due to localized persistent structures, we use the same lag length for the AR model. The use of more lags led to better results for the AR model but did not improve the identification of localized pockets; thus, the models in the analysis were set to be comparable and parsimonious.
\clearpage
\begin{table}[t!] 
\caption{\footnotesize{The table reports the forecasting performance measured according to the root mean square error (RMSE) and mean absolute error (MAE) for our time-varying extended Wold decomposition (TV-EWD) model, time-varying autoregressive (TV-AR) model, autoregressive (AR) model and extended Wold decomposition (EWD) model relative to benchmark models. The three panels report results for three cases. Panel A reports the results for the full sample. Panel B reports the results of the time-varying models for the periods defined by the pockets. The performance of the models in Panel A and Panel B is reported relative to that of the random walk benchmark model for 1-, 2-, 6- and 12-month forecasts. Panel C reports the results for the periods identified as pockets relative to the AR model. We only use pockets that have less than a 5\% chance of being spurious. The sampling distributions used to determine outliers are from an AR residual bootstrap design.}}
\begin{center} 
\footnotesize 
\begin{tabular}{l cccccccccccc} 
\toprule 
& \multicolumn{2}{c}{1-month-ahead} & & \multicolumn{2}{c}{2-months-ahead} & & \multicolumn{2}{c}{6-months-ahead} & & \multicolumn{2}{c}{12-months-ahead}\\ 
\cmidrule(r){2-3} \cmidrule(r){5-6} \cmidrule(r){8-9} \cmidrule(r){11-12} 
& RMSE & MAE & & RMSE & MAE & & RMSE & MAE & & RMSE & MAE\\ 
\cmidrule(r){2-12} 
\multicolumn{12}{c}{Panel A: Full sample (relative to RW)}\\
\cmidrule(r){2-12} 
AR             & 0.893 & 0.931 & & 0.830 & 0.889 & & 0.754 & 0.869 & & 0.870 & 1.087\\ 
EWD              & 0.911 & 0.941 & & 0.838 & 0.885 & & 0.715 & 0.745 & & 0.783 & 0.752\\ 
TV-AR        & 0.987 & 1.038 & & 0.790 & 0.858 & & 0.700 & 0.707 & & 0.803 & 0.738\\ 
TV-EWD & 0.845 & 0.833 & & 0.760 & 0.792 & & 0.655 & 0.689 & & 0.753 & 0.720\\
\cmidrule(r){2-12} 
\multicolumn{12}{c}{Panel B: Pockets (relative to RW)}\\
\cmidrule(r){2-12} 
TV-AR &  0.807 & 0.760 & &0.819 &  0.838 & &0.599 & 0.604 & &0.541 & 0.435 \\
TV-EWD & 0.789 & 0.723 & & 0.620 & 0.618 & &0.437 & 0.456 && 0.460  & 0.424 \\
\cmidrule(r){2-12} 
\multicolumn{12}{c}{Panel C: Pockets (relative to AR)} \\
\cmidrule(r){2-12} 
TV-AR & 0.856 & 0.873 & &0.735 & 0.751 && 0.582 & 0.586 & &0.541 & 0.522 \\
TV-EWD &  0.817 & 0.809 & &0.668  &0.637 & &0.537 & 0.557  &&0.467 & 0.437 \\
\bottomrule 
\end{tabular} 
\end{center} 
\label{tab:inflation} 
\end{table}

We divide the observed time series into two parts, namely, in-sample $x_{1,T},\ldots,x_{\text{ins},T}$ and out-of-sample $x_{\text{ins}+1,T},\ldots,x_{T,T}$ observations. Then, we use the in-sample data to fit the models and determine the out-of-sample predictive performance using the RMSE and MAE loss functions. Using the first 645 observations for the in-sample observations, we are left with 136 out-of-sample observations. The forecast horizons considered are $h\in\{1,2,6,12\}$ months ahead. The results of the forecast, expressed as the mean of the loss functions relative to the RW model, are shown in Table \ref{tab:inflation}.
\begin{figure}[ht!]
            \begin{center}
                \includegraphics[scale=0.4]{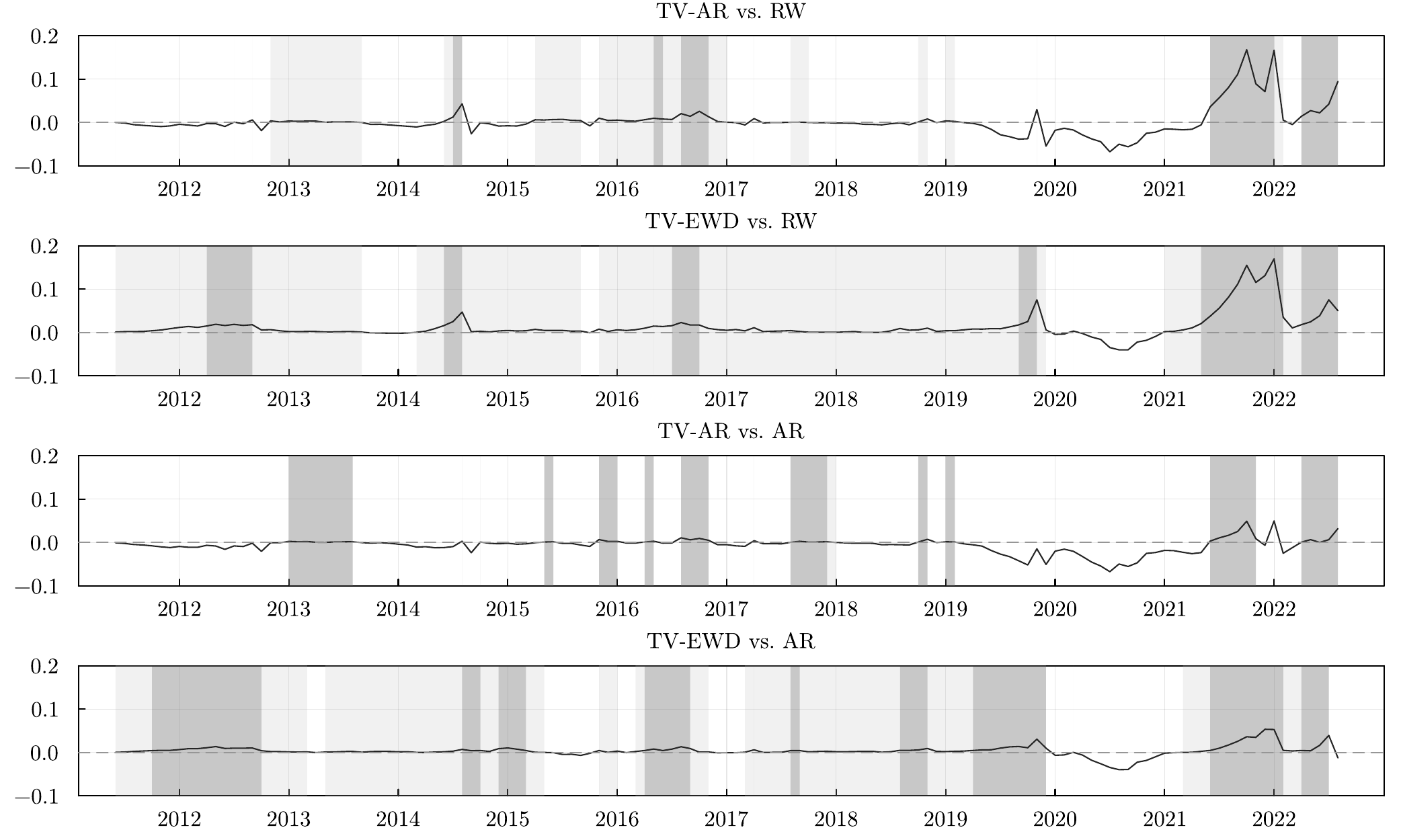}
            \end{center}
            \caption{\footnotesize{Local inflation predictability. The four panels show plots of one-sided kernel estimates of the fitted squared forecast error differential $\widehat{\text{SED}}^h_t$ from time-varying autoregressive (TV-AR) and time-varying extended Wold decomposition (TV-EWD) models versus those of random walk (RW) and autoregressive (AR) models, respectively, for $h=1$. The shaded areas represent periods where $\widehat{\text{SED}}^h_t>0$, with the dark areas representing pockets with less than a 5\% chance of being spurious and light areas representing pockets with more than a 5\% chance of being spurious. The sampling distributions used to determine spuriousness are from an AR residual bootstrap design.}}
            \label{fig:pocketsinflation}
\end{figure} 

The results show that the TV-EWD model provides the best forecasting performance at all forecast horizons. This advantage increases as the forecast horizon increases, suggesting that accurate identification of the rich persistent structures within the inflation time series is important for longer-term forecasting. The performance of the EWD and TV-AR models also improves as the horizon length increases, but the combination of the two approaches, namely, our TV-EWD models, performs best.

Importantly, as shown in panels B and C in Table \ref{tab:inflation}, which provide the results for pockets identified by both the TV-AR and TV-EWD models, our models provide better localization results than both the RW and AR benchmark models. 
Both time-varying models can be used to identify periods with higher losses, which are further illustrated in Figure \ref{fig:pocketsinflation}, which shows plots of one-sided nonparametric kernel estimates of $\widehat{\text{SED}}^h_t$ against time for each of the models. 
The shaded areas represent periods identified as pockets of predictability, with the dark areas representing pockets with less than a 5\% chance of being spurious and the light areas representing pockets with more than a 5\% chance of being spurious.

We conclude that pockets of inflation predictability, identified via a local persistence decomposition approach applied to the inflation series, can be used to increase inflation forecasting accuracy through the use of newly discovered information.

\subsection{Time-Varying Pockets of Persistence in the Volatility of U.S. Stocks}

The second important time series with a potentially interesting structure for our analysis is volatility. Volatility is one of the most important measures in finance, as it captures fluctuations in asset prices and hence their risk. We use daily volatility data\footnote{The realized volatility is computed as the sum of the squared logarithmic 5-minute returns of the sample for each day.} for all stocks listed in the S\&P 500 index from July 5, 2005, to August 31, 2018, from TickData, so we work with data for the 496 stock returns over 3278 days.
\begin{figure}[h!]
\begin{center}
            \includegraphics[scale=0.5]{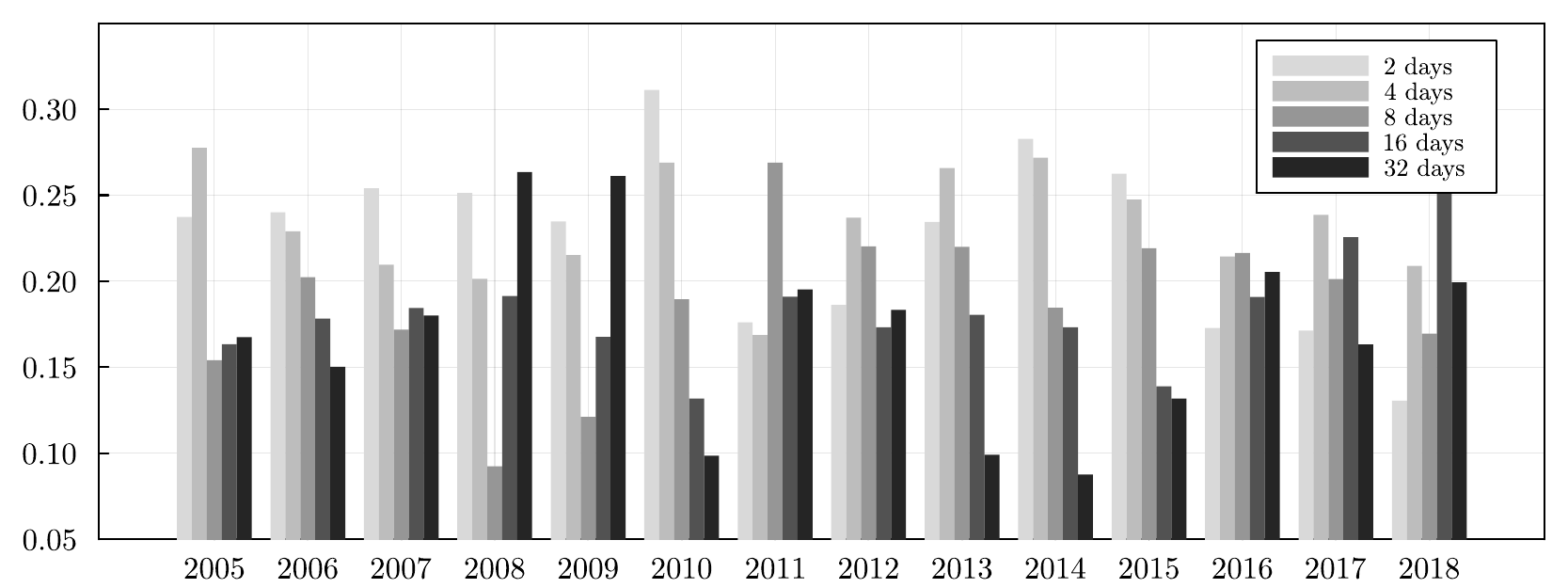}
    \end{center}
        \caption{\footnotesize{Time-varying persistent structures in the series of the realized volatility of Agilent Technologies, Inc. The plots show ratios of $\widehat{\beta}^{j}(t/T,1)/\sum_j \widehat{\beta}^{j}(t/T,1)$, with $j$ corresponding to the 2-, 4-, 8-, 16-, and 32-day persistence of shocks over the period from July 2005 until August 2018.}}
        \label{fig:volatility}       
\end{figure} 

Again, we start by illustrating the persistent structures within the data. Since our sample contains 496 stocks, we examined the first available stock (in alphabetical order): Agilent Technologies, Inc. Note that we examined other stocks and found that the persistent structures within these series are similarly rich to the series discussed in this analysis. Figure \ref{fig:volatility} shows the average $\widehat{\beta}^{j}(t/T,k)/\sum_j \widehat{\beta}^{j}(t/T,k)$ ratio of the $j$ scales, representing 2-, 4-, 8-, 16-, and 32-day persistence of shocks for the sample for each year for $k=1$. Thus, for each year, we can see the average contribution of the shocks to the volatility series. The reason we look at the averages is that the daily sample contains rich dynamics that are difficult to visualize. Moreover, the use of aggregate information for one year strongly supports our objective.

Specifically, we observe some periods that are driven mostly by transitory shocks of up to 4 days, such as 2005, 2010 and 2014, as well as periods that are driven by more persistent shocks, such as 2008, 2011 and 2016--2018.

Overall, the results demonstrate the richness of the persistence dynamics of volatility series. While it is important to capture the time variations of the dependent structures in the series, it is also crucial to capture the smoothly changing persistent structures.

\subsubsection{Forecasting Volatility}

Finally, we use the time-varying persistent structure to build a more accurate forecasting model for volatility. We use the EWD and TV-AR models as natural comparison methods for our TV-EWD model, as these models consider two types of time series representations based on different levels of persistence (EWD) and time variations (TV-AR). In addition, we use the heterogeneous AR (HAR) model proposed by \cite{corsi2009}, which is a standard benchmark in the literature. The HAR model is defined as $x_t=\beta_0+\beta_1x_{t-1}+\beta_5\overline{x}_{t-5}+\beta_{22}\overline{x}_{t-22}+\epsilon_t$, where $\overline{x}_{t-m}=1/m\sum_{h=0}^{m- 1}x_{t-m}$, which is the average of the weekly and monthly realized volatility. This model captures the persistent structures in financial volatility series well. We also consider the extension of the TV-HAR model \citep{wang2017time}:
\begin{equation}
x_{t,T}=\beta_0\left(\frac{t}{T}\right)+\beta_1\left(\frac{t}{T}\right)x_{t-1,T}+\beta_5 \left(\frac{t}{T}\right)\overline{x}_{t-5,T}+\beta_{22} \left(\frac{t}{T}\right)\overline{x}_{t-22,T} + \epsilon_t,
\label{eq:tvar_model}
\end{equation}

We compare the out-of-sample forecasting performance of our TV-EWD model with that of the comparison models based on the realized volatilities of all S\&P 500 companies available in the sample. We estimate the model parameters based on the information set containing the first 1000 observations and use the remaining data for 1-, 5- and 22-day out-of-sample tests. To explore the changing behavior of the data, we also examine different periods to determine how sample specific the results are. The richer the localized structures in the data are, the larger the expected performance gains. According to the literature, in addition to the aggregate results for the entire out-of-sample period from August 2009 to August 2018, we look at two specific periods, August 2009 to August 2012 (at the beginning of the out-of-sample period) and August 2016 to August 2017. Finally, we identify pockets of predictability that we expect to further improve the results compared with those of these three selected samples.

\clearpage
\begin{table}[t!]
\footnotesize 
\caption{\footnotesize{The table reports the forecasting performance, measured according to the root mean square error (RMSE) and mean absolute error (MAE), for our time-varying extended Wold decomposition (TV-EWD), time-varying autoregression (TV-AR), heterogeneous autoregressive (TV-HAR) and extended Wold decomposition (EWD) models. All errors are relative to those of the popular heterogeneous autoregressive (HAR) model over forecast horizon $h$. The figures show the median values for all 496 companies in the sample, calculated based on forecasts for the three periods reported in the first three panels (A, B and C). The last panel (D) reports the results for the periods identified as pockets of predictability. We only use pockets that have less than a 5\% chance of being spurious. The sampling distributions used to determine spuriousness come from an AR residual bootstrap design.}}
\begin{center} 
\begin{tabular}{llccccccc} 
\toprule 
& & \multicolumn{3}{c}{RMSE} & & \multicolumn{3}{c}{MAE} \\ 
\cmidrule(r){3-5} \cmidrule(r){7-9}
& & $h=1$ & $h=5$ & $h=22$ & & $h=1$ & $h=5$ & $h=22$ \\ 
\cmidrule(r){1-5} \cmidrule(r){7-9}
 Panel A: &  EWD &   1.049 & 1.009 & 0.944 & & 0.991 & 0.925 & 0.883 \\
August 2009 - August 2018& TV-AR & 1.033 & 1.09 &  0.944 & & 1.014 & 1.055 & 0.888\\
& TV-HAR & 0.961 & 0.946 & 0.832 & & 0.935 & 0.905 & 0.786\\
& TV-EWD & 0.970 &  0.972 & 0.821 & & 0.924 & 0.901 & 0.760\\
\cmidrule(r){1-5} \cmidrule(r){7-9}
 Panel B: & EWD &   1.051 & 0.981 & 0.974 & & 0.989 & 0.896 & 0.838 \\
August 2009 - August 2012 & TV-AR & 1.061 & 1.218 & 1.084 & & 1.047 & 1.171 & 0.955\\
& TV-HAR & 0.922 & 0.895 & 0.826 & & 0.891 & 0.857 & 0.725\\
& TV-EWD & 0.922 & 0.892 & 0.749 & & 0.872 & 0.823 & 0.629\\
\cmidrule(r){1-5} \cmidrule(r){7-9}
 Panel C: & EWD &   1.002 & 1.021 & 1.032 & & 0.937 & 0.931 & 1.000\\
August 2016 - August 2017  & TV-AR & 1.010 &  1.033 & 1.012 & & 1.006 & 1.049 & 1.002\\
& TV-HAR & 0.957 & 0.958 & 0.968 & & 0.943 & 0.923 & 0.964\\
& TV-EWD & 0.959 & 0.969 & 1.001 & & 0.908 & 0.901 & 0.931\\
\cmidrule(r){1-5} \cmidrule(r){7-9}
  Panel D:& TV-AR &   0.915 & 0.841 & 0.664 && 0.877 & 0.802 & 0.664 \\
   Pockets & TV-HAR & 0.903 & 0.834 & 0.676 && 0.867 & 0.796  &0.665 \\
   & TV-EWD &         0.857 & 0.785 & 0.611 && 0.815 & 0.719 & 0.592 \\
\bottomrule 
\end{tabular} 
\end{center} 
\label{tab:RVs} 
\end{table}

For estimation and forecasting via the TV-EWD model, we use the procedures described in Section \ref{sec:forecasting} with $J=\{5,5,7\}$ scales and $\{2,5,15\}$ autoregressive lags for $h=\{1,5,22\}$ forecasts. These choices are optimized for the best losses. Note that the choice of higher-order lags in the AR model inherently results in better forecasts with an increasing horizon. The kernel width of 0.2 minimizes the MSEs of the forecasts. We compare the forecasting performance of the different models using the MAE and RMSE loss functions relative to the benchmark HAR model.

\begin{figure}[t!]
            \begin{center}
                \includegraphics[scale=0.4]{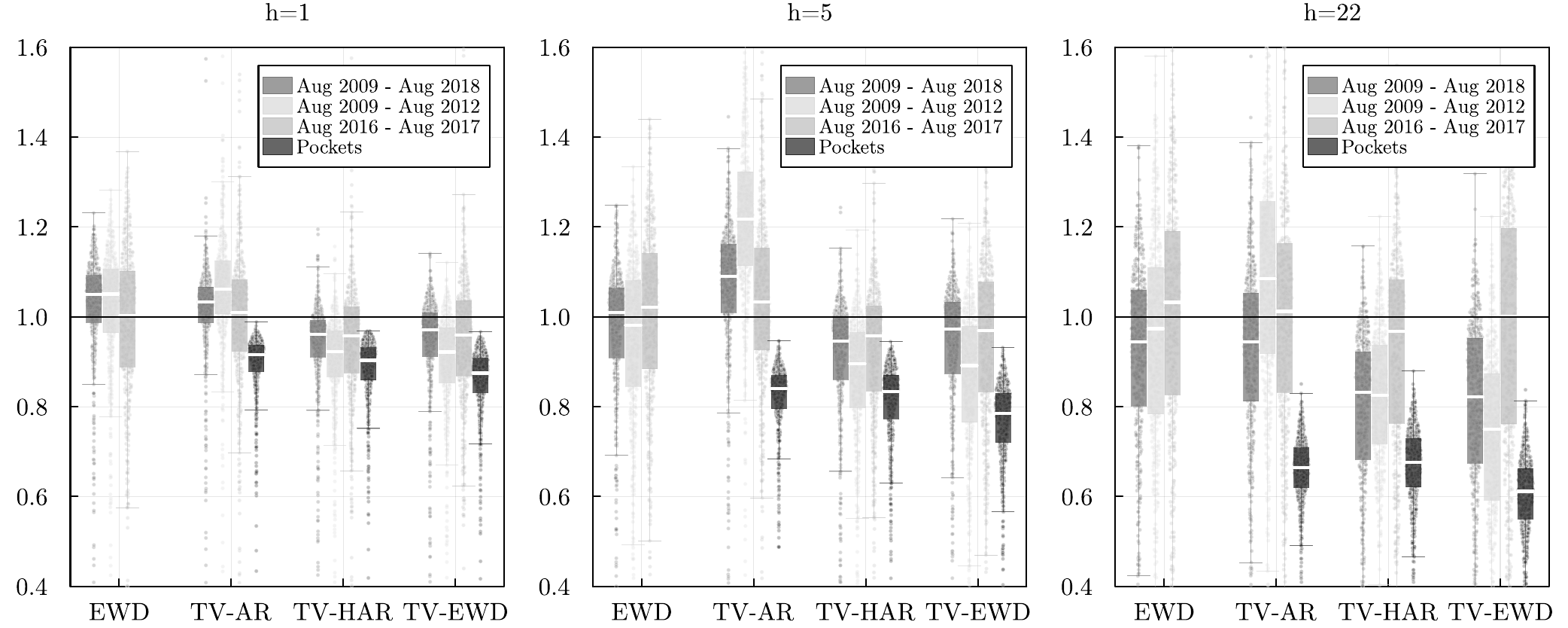}
            \end{center}
            \caption{\footnotesize{Root mean square error of our TV-EWD model compared with the extended Wold decomposition (EWD), time-varying autoregression (TV-AR) and time-varying heterogeneous autoregression (TV-HAR) models. All the errors are relative to those of the HAR model of \cite{corsi2009} over $h=1$ (left), $h=5$ (middle) and $h=22$ (right). The box plots show the RMSE values for all 496 companies in the sample, computed based on the forecasts for the four periods represented by the four colors, including the periods identified as pockets of predictability. We only use pockets that have less than a 5\% chance of being spurious. The sampling distributions used to determine spuriousness come from an AR residual bootstrap design.}}
            \label{rmse_rv}
\end{figure} 

\begin{figure}[t!]
            \begin{center}
                \includegraphics[scale=0.4]{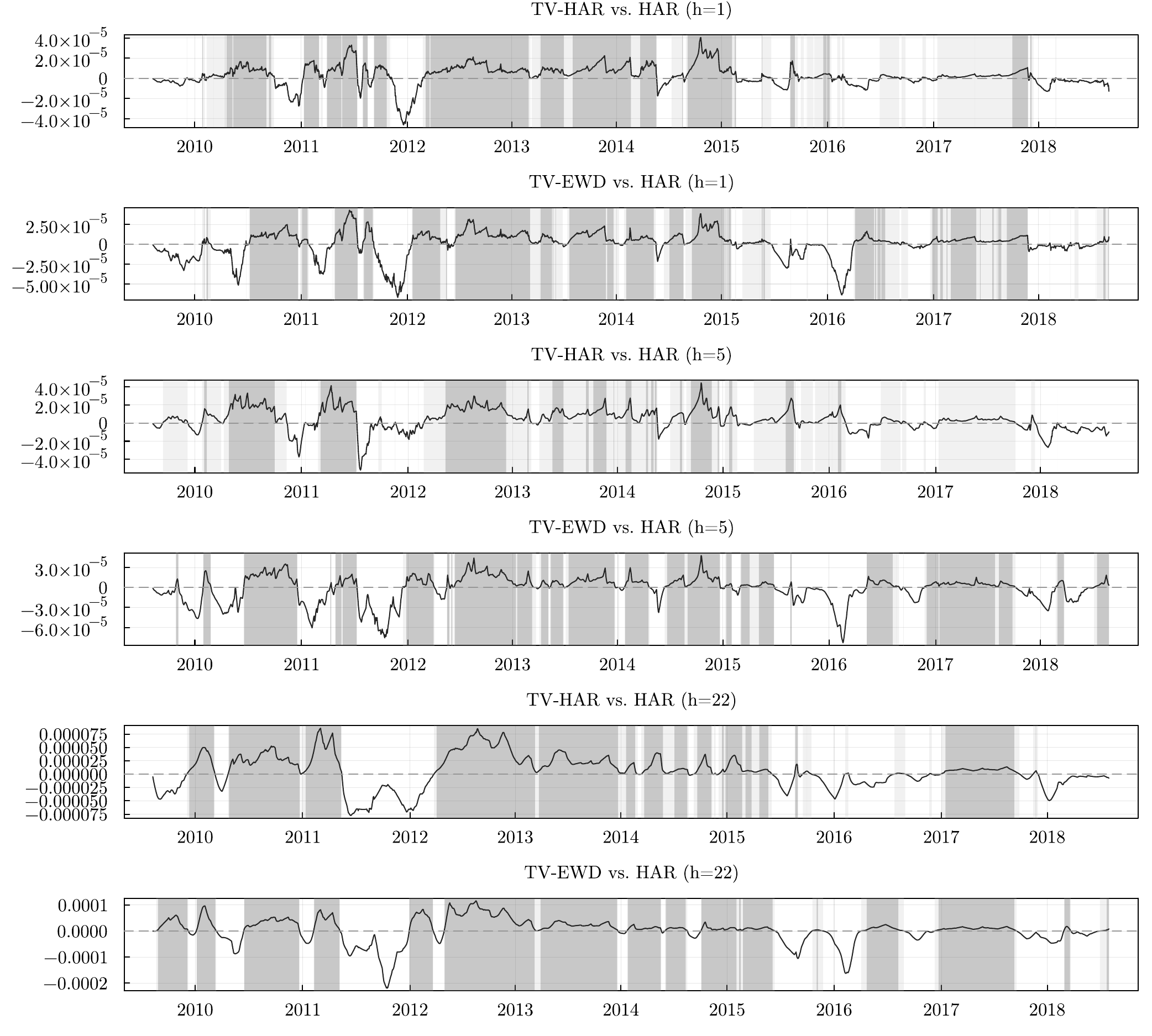}
            \end{center}
            \caption{\footnotesize{Local volatility prediction for Agilent Technologies Inc. The panels show one-sided kernel estimates of the fitted squared forecast error differential $\widehat{\text{SED}}^h_t$ from time-varying heterogeneous autoregressive (TV-HAR) and time-varying extended Wold decomposition (TV-EWD) models versus a heterogeneous autoregressive (HAR) model for $h\in\{1,5,22\}$. The shaded areas represent periods where $\widehat{\text{SED}}^h_t>0$, with the dark areas representing pockets that have less than a 5\% chance of being spurious and the light areas representing pockets that have more than a 5\% chance of being spurious. The sampling distributions used to determine spuriousness are from an AR residual bootstrap design.}}
            \label{fig:pocketsvolatility}
\end{figure} 

As the above results were obtained based on a large sample of 496 stocks, we provide the results in Table \ref{tab:RVs}, which reports the median estimates for all stocks. In addition, we provide box plots showing the errors for all stocks in Figure \ref{rmse_rv} (RMSE). \footnote{Note that the plots in Figure \ref{mae_rv} in Appendix \ref{sec:additional} show the absolute errors, which essentially provide the same information; therefore, we do not present them in the main text of the paper.} According to the results in Table \ref{tab:RVs}, the TV-EWD model outperforms all the other models for the different forecast horizons and samples. In particular, this result holds for the median of the errors computed for the 496 stocks in the sample, except for the RMSE for the first period, August 2009--August 2012, where the TV-HAR model consistently produces better forecasts. The results in Figure \ref{rmse_rv} (RMSE), which shows more detailed results for all stocks in box plots, confirm that the TV-EWD model produces much better forecasts than do all the other models for most of the stocks considered for all the horizons and samples considered.

Looking at the results in more detail, it is important to note that we are comparing several different approaches. First, a popular HAR model captures the unconditional persistent structures with 22 lags in the AR model, whereas the EWD model provides better results at longer horizons through the identification of more precise persistent structures. This result is consistent with the findings of \cite{ortu2020}, although they used only a single time series, and our results hold for various volatile stocks.

Second, and more importantly, adding temporal variations to AR models significantly improves the results, as time-varying parameters capture the dynamics in the data. In particular, the TV-HAR model provides significantly better forecasts compared with those of the other models. Finally, when the persistent structure is allowed to vary smoothly over time, as in our TV-EWD model, we note further improvements in the forecasts. The ability to appropriately incorporate changing persistent structures in the data into the model represents an advantage of the TV-EWD model, especially at longer horizons.

Interestingly, in the much quieter period from August 2016--August 2017, where we do not identify very heterogeneous persistent structures, the complex TV-EWD model performs similarly to the HAR and TV-HAR models in terms of the RMSE, although it still has the best results in terms of the MAE.

To quantify the benefits of localization more precisely, the pocket results in panel D of Table \ref{tab:RVs}, as well as the corresponding box plots in Figure \ref{rmse_rv} (RMSE), clearly show further increases in accuracy for the periods identified as pockets of predictability with our model. This result holds robustly for the 496 volatile stocks and thus provides evidence of the significant improvements achieved via the identification of pockets of predictability using a time-varying regression approach, particularly our decomposition strategy. 
Figure \ref{fig:pocketsvolatility} further illustrates the pockets of predictability for the company Agilent Technologies, Inc., identified by localizing the persistent structures via the decomposition strategy. 
While this serves as an example and similar patterns are observed for all the other stocks, this result demonstrates that the TV-EWD model identifies different periods than the TV-HAR model, which represents the best comparative model. Furthermore, these pockets increase the forecast accuracy through the identification of local persistent structures.

\section{Conclusion}
 \label{sec:conclusion}

We develop a novel framework for capturing smoothly evolving persistent structures in economic time series, allowing for the identification of rich, time-varying dynamics that are often overlooked by traditional models. By decomposing series into locally stationary components with heterogeneous persistence, our approach reveals structures that are otherwise hidden and provides a more nuanced understanding of shock propagation over time.

The model was applied to inflation and stock market volatility series and could identify distinct pockets of predictability, namely, periods in which persistent components dominate and forecasting accuracy increases. These findings demonstrate the practical value of our method for both empirical modeling and policy-relevant forecasting. More broadly, our results underscore the importance of accounting for time-varying persistence when analyzing economic and financial data.

\bibliography{ext_Wold}
\bibliographystyle{chicago}


\newpage

\appendix
\section{Appendix: Locally Stationary Processes}
\label{sec:app}

\begin{assumption}
\label{propLSP}
Locally stationary processes \cite{dahlhaus2009}: Let the sequence of stochastic processes $x_{t,T}$, ($t=1,\cdots,T$) be called a locally stationary process if $x_{t,T}$ has a representation
\begin{equation}
x_{t,T} = \sum_{h=-\infty}^{+\infty} \alpha_{t,T} \left(h\right) \epsilon_{t-h}
\end{equation}
satisfying the following conditions:
\begin{equation}
\sup_{t,T}=\vert  \alpha_{t,T} (h) \vert \le \frac{K}{l(h)},
\end{equation}
where $l(h)$ for some $\kappa>0$ is defined as:
\begin{equation}
  l(h):=
    \begin{cases}
      1 & \vert h\vert \le1\\
      \vert h\vert\log^{1+\kappa}\vert h\vert & \vert h\vert >1
    \end{cases}       
\end{equation}
and $K$ is not dependent on $T$. In addition, there exist functions $\alpha(\cdot,h):(0,1]\rightarrow \R $ with
\begin{equation}
\sup_{t=1,\ldots,T}=\left\vert  \alpha\left(\frac{t}{T},h\right) \right\vert \le \frac{K}{l(h)},
\end{equation}
\begin{equation}
\sup_h \sum^n_{t=1} \left\vert \alpha_{t,T} \left(h\right)  -  \alpha\left( \frac{t}{T},h \right) \right\vert \le K,
\end{equation}
\begin{equation}
V(\alpha(\cdot,h))\le \frac{K}{l(h)},
\end{equation}
where $V(\cdot)$ denotes the total variation in $[ 0,1]$, $\epsilon\sim$ iid, $\E\epsilon_t \equiv 0$, and $\E\epsilon_t^2 \equiv 1$. We also assume that all the moments of $\epsilon_t$ exist.
\end{assumption}

\newpage
\begin{assumption} 
\label{propLSP_beta}
Let the sequence of stochastic processes $x_{t,T}$, ($t=1,\cdots,T$) be called a locally stationary process if $x_{t,T}$ has a representation
\begin{equation}
x_{t,T}=\sum_{j=1}^{+\infty} \sum_{k=0}^{+\infty} \beta_{t,T}^{\{j\}}(k) \epsilon_{t-k2^j}^{\{j\}},
\end{equation}
satisfying the following conditions $\forall j$:
\begin{equation}
\sup_{t,T}=\vert  \beta_{t,T}^{\{j\}} (h) \vert \le \frac{K}{l(h)},
\end{equation}
where $l(h)$ for some $\kappa>0$ is defined as:
\begin{equation}
  l(h):=
    \begin{cases}
      1 & \vert h\vert \le1\\
      \vert h\vert\log^{1+\kappa}\vert h\vert & \vert h\vert >1
    \end{cases}       
\end{equation}
and $K$ is not dependent on $T$. In addition, there exist functions $\beta^{\{j\}}(\cdot,h):(0,1]\rightarrow \R $ with
\begin{equation}
\sup_{t=1,\ldots,T}=\left \vert  \beta^{\{j\}}\left(\frac{t}{T},h\right) \right\vert \le \frac{K}{l(h)},
\end{equation}
\begin{equation}
\sup_h \sum^n_{t=1} \left\vert \beta^{\{j\}}_{t,T} \left(h\right)  -  \beta^{\{j\}}\left( \frac{t}{T},h \right) \right\vert \le K,
\end{equation}
\begin{equation}
V(\alpha(\cdot,h))\le \frac{K}{l(h)},
\end{equation}
where $V(\cdot)$ denotes the total variation in $[ 0,1]$, $\epsilon\sim$ iid, $\E\epsilon_t \equiv 0$, and $\E\epsilon_t^2 \equiv 1$. We also assume that all the moments of $\epsilon_t$ exist.
\end{assumption}

\section{Appendix: Obtaining the Wold Representation for the TV-AR$(p)$ Model}
\label{sec:app_wold}

The Wold representation for a time series is found from an AR specification, through which a time series $x_t$ is expressed as
\begin{equation}
x_t = \sum_{j=0}^{\infty} \alpha_j \varepsilon_{t-j}
\end{equation}
where \( \varepsilon_t \) is white noise (i.i.d. with mean zero and variance \( \sigma^2 \)), and \( \alpha_j \) are the MA coefficients for \(j = 0, 1, 2, \ldots \), which, in principle, requires finding an infinite number of parameters in the data. With a finite number of observations in $(x_1,x_2,\ldots,x_T)$, we need to make additional assumptions about the nature of these parameters, such as that they can be expressed by rearranging the AR($p$) model
\begin{equation}
x_t=\phi_1x_{t-1}+\phi_2x_{t-2}+\ldots+\phi_px_{t-p}+\epsilon_t
\end{equation}
in lag operator notation, with $L$ being the lag operator such that $Lx_t=x_{t-1}$
\begin{equation}
x_t=\frac{1}{1-\phi_1L-\phi_2L^2-\ldots-\phi_pL^p}\epsilon_t
\end{equation}
This equation can be expanded as
\begin{equation}
x_t = (1 + \alpha_1 L + \alpha_2 L^2 + \alpha_3 L^3 + \cdots) \varepsilon_t
\end{equation}
where the \( \alpha \)-coefficients are MA coefficients representing the Wold decomposition.

To obtain the Wold representation, we compute the \( \alpha \) coefficients recursively. For an AR(1) model, these coefficients are simply powers of \( \phi \), i.e., $\alpha_j = \phi^j$. For higher-order AR models, recursion is used:
\[
\alpha_j = \sum_{k=1}^{\min(j, p)} \phi_k \alpha_{j-k} \quad \text{for } j \geq 1
\]
where $\alpha_0 = \sigma_{\epsilon_t}$. This recursive formula provides the MA coefficients up to any desired lag length $j$.

For the time-varying parameter autoregressive (TV-AR) process of order $p$, such as
\begin{equation}
x_t = \phi_{1,t} x_{t-1} + \phi_{2,t} x_{t-2} + \cdots + \phi_{p,t} x_{t-p} + \varepsilon_t
\end{equation}
where \( \{\phi_{k,t}\} \) are the time-varying autoregressive coefficients and \( \varepsilon_t \sim \text{i.i.d. } (0, \sigma^2) \) is white noise, we need to adjust the recursion to account for parameters that change over time so that the Wold representation will also change over time. This is possible by assuming localization at a given time point, as discussed in the main text of this paper:
\begin{equation}
x_t = \sum_{j=0}^{\infty} \alpha_{j,t} \varepsilon_{t-j}
\end{equation}
where \( \alpha_{j,t} \) are the time-varying MA coefficients, which depend on the values of \( \phi_{k,t} \) at each time step.

To obtain the time-varying Wold representation, we compute the \( \alpha \) coefficients recursively at each time $t$. That is, for each time \(t \), we set \( \alpha_{0,t} = \sigma_{\epsilon_t} \) (representing the effect of the current 
step). 
Then, the following terms are calculated:
   \[
   \alpha_{j,t} = \sum_{k=1}^{\min(j, p)} \phi_{k,t} \alpha_{j-k,t} \quad \text{for } j \geq 1
   \]
   With this recursion, \( \alpha_{j,t} \) can capture the time-varying dynamics of \( \phi_{k,t} \). In practice, the recursion is truncated at a maximum lag length \( J \) chosen according to the desired accuracy of the approximation.

\section{Appendix: Additional Plots and Tables}
\label{sec:additional}

\begin{figure}[!t]
            \begin{center}
                \includegraphics[scale=0.4]{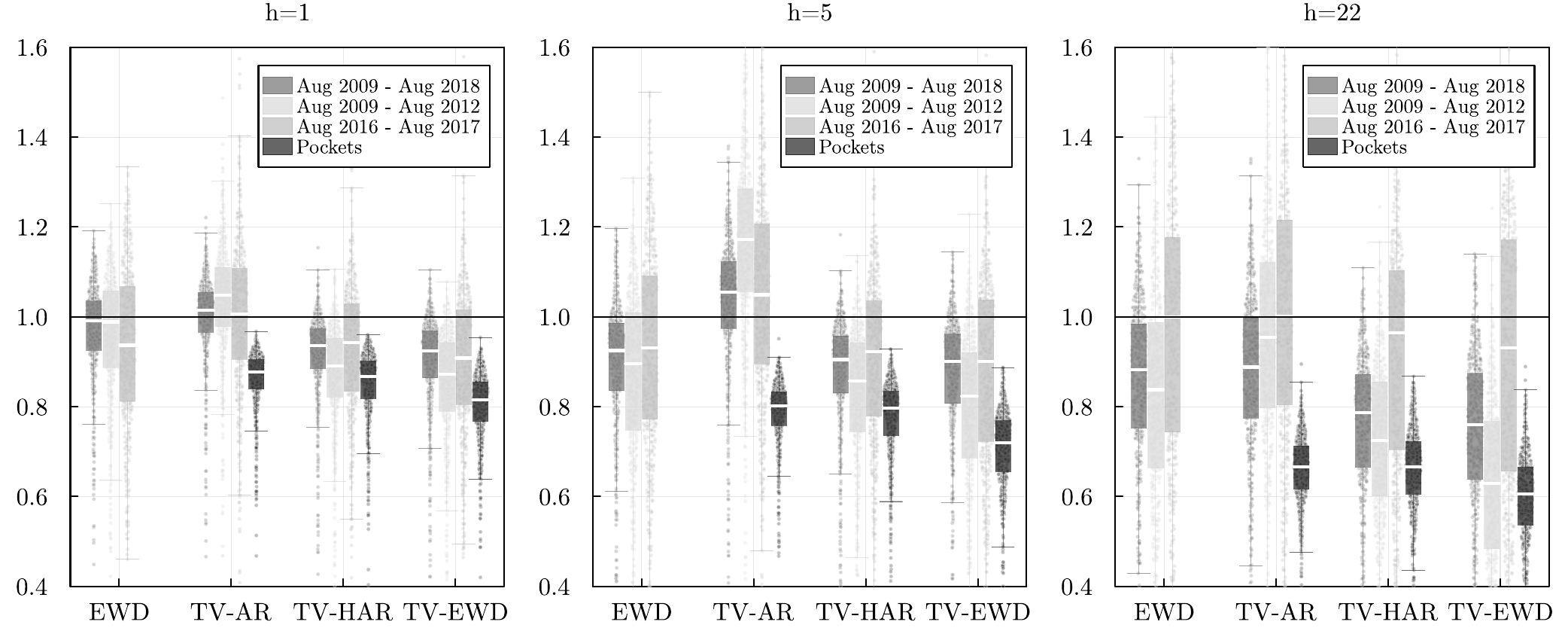}
            \end{center}
            \caption{\footnotesize{Mean absolute error of our TV-EWD model compared with the extended Wold decomposition (EWD), time-varying autoregression (TV-AR) and time-varying heterogeneous autoregression (TV-HAR) models. All the errors are relative to the HAR model of \cite{corsi2009} over $h=1$ (left), $h=5$ (middle) and $h=22$ (right). The box plots show the RMSE values for all 496 companies in the sample, computed based on the forecasts for the four periods represented by the four colors, including the periods identified as pockets of predictability. We use only pockets that have less than a 5\% change of being spurious. The sampling distribution used to determine spuriousness comes from an AR residual bootstrap design.}}
            \label{mae_rv}
\end{figure} 

\end{document}